\documentclass[3p,number]{elsarticle}

\usepackage{color}
\usepackage{amsfonts}
\usepackage[boxed,ruled,vlined,linesnumbered,scleft]{algorithm2e}
\usepackage{hyperref}
\usepackage{graphicx}
\usepackage{epstopdf}
\usepackage{amssymb}

\newcommand{\eclause}{{\unitlength 2.3mm \,\framebox(1,1){}\,}}

\newcommand{\ignore}[1]{}
\newtheorem{theorem}{Theorem}
\newtheorem{corollary}{Corollary}
\newtheorem{definition}{Definition}
\newtheorem{lemma}{Lemma}
\newenvironment{proof}{\medskip\par\noindent{\sc Proof:}}{\hfill\rule{2mm}{2mm
}\medskip\par}
\newcommand{\E}{\mathrm{E}}

\begin{document}
\begin{frontmatter}

\title{Scale-Free Random SAT Instances%
  \tnoteref{thanks}}
\tnotetext[thanks]{Research  partially supported by the EU H2020 Research and Innovation Program under the LOGISTAR project (Grant Agreement No. 769142) and the MINECO-FEDER projects RASO (TIN2015-71799-C2-1-P) and TASSAT~3 (TIN2016-76573-C2-2-P).}

\author{Carlos Ans\'otegui}
\ead{carlos@diei.udl.cat}
\address{DIEI, UdL, Jaume II 69,  Lleida, Spain}
\author{Maria Luisa Bonet}
\ead[url]{http://www.lsi.upc.edu/~bonet}\ead{bonet@lsi.upc.edu}
\address{LSI, UPC, J.~Girona 1--3,  Barcelona, Spain}
\author{Jordi Levy}
\ead[url]{http://www.iiia.csic.es/~levy}
\ead{levy@iiia.csic.es}
\address{IIIA-CSIC, Campus UAB,  Bellaterra, Spain}

\begin{abstract}
We focus on the random generation of SAT instances that have properties similar to real-world instances. It is known that many industrial instances, even with a great number of variables, can be solved by a clever solver in a reasonable amount of time. This is not possible, in general, with classical randomly generated instances. We provide a different generation model of SAT instances, called \emph{scale-free random SAT instances}. It is based on the use of a non-uniform probability distribution $P(i)\sim i^{-\beta}$ to select variable $i$, where $\beta$ is a parameter of the model. This results into formulas where the number of occurrences $k$ of variables follows a power-law distribution $P(k)\sim k^{-\delta}$ where $\delta = 1 + 1/\beta$. This property has been observed in most real-world SAT instances. For $\beta=0$, our model extends classical random SAT instances.

We prove the existence of a SAT-UNSAT phase transition phenomenon for scale-free random 2-SAT instances with $\beta<1/2$ when the clause/variable ratio is $m/n=\frac{1-2\beta}{(1-\beta)^2}$.  We also prove that scale-free random k-SAT instances are unsatisfiable with high probability when the number of clauses exceeds $\omega(n^{(1-\beta)k})$. 
The proof of this result suggests that, when $\beta>1-1/k$, the unsatisfiability of most formulas may be due to small cores of clauses.
Finally, we show how this model will allow us to generate random instances similar to industrial instances, of interest for testing purposes.
\end{abstract}
\end{frontmatter}

\section{Introduction}

Over the last 20 years, SAT solvers have experienced a great
improvement in their efficiency when solving practical SAT
problems. This is the result of some techniques like conflict-driven
clause learning (CDCL), restarting and clause deletion policies. The
success of SAT solvers is surprising, if we take into account that SAT
is an NP-complete problem, and in fact, a big percentage of formulas
need exponential size resolution proofs to be shown
unsatisfiable. This has led some researchers to study what is the
nature of real-world or industrial SAT instances that make them easy
in practice. Parallelly, most theoretical work on SAT has focused on uniform randomly selected
instances. Nevertheless, nowadays we know that most industrial
instances share some properties that are not present in most (uniform
randomly-choosen) SAT formulas. It is also well-known that solvers
that perform well on industrial instances, do not perform well on
random instances, and vice versa. Therefore, a new
theoretical paradigm that describes the distribution of industrial
instances is needed. Not surprisingly, generating random instances that are more
similar to real-world instances is described as one of the ten grand
challenges in
satisfiability~\cite{tenchallenges1,tenchallenges2,tenchallenges3,tenchallenges4}.

Over the last 10 years, the analysis of the industrial SAT instances
used in SAT solvers competitions has allowed us to have a clear image
of the structure of real-world instances. \citet{backdoors} proved
that industrial instances contain a small number of variables (called
the \emph{backdoor} of the formula) that, when instantiated, make the
formula easy to solve. \citet{AAAI08} showed that industrial instances
have a smaller tree-like resolution space complexity (also called
\emph{hardness}~\cite{kullmann}) than randomly generated instances with the same
number of variables. \citet{CP09} proved that most industrial
instances, when represented as a graph, have a scale-free
structure. This kind of structure has also been observed in other
real-world networks like the World Wide Web, Internet, some social
networks like papers co-authorship or citation, protein interaction
network, etc. \citet{SAT12,JAIR19} show that these graph representations of
industrial instances exhibit a very high modularity. Modularity has
been shown to be correlated with the runtime of CDCL SAT
solvers~\cite{NewshamGFAS14}, and has been used to improved the
performance of some solvers~\cite{SAT15,SonobeKI14,DMartinsML13}. It
is also known that these graph representations are
self-similar~\cite{IJCAR14} and eigen-vector centrality is correlated
with the significance of variables~\cite{pagerank}.

Defining a model that captures all these properties observed in
industrial instances is a hard task. Here, we focus on the
scale-free structure. We will define a model and propose a
random generator for \emph{scale-free SAT formulas}, extending our work
presented in CCIA'07~\cite{CCIA07}, CCIA'08~\cite{CCIA08}, 
IJCAI'09~\cite{IJCAI09} and CP'09~\cite{CP09}. This model is parametric in the
size $k$ of clauses and an exponent $\beta$. Formulas are sets of $m$ independently sampled clauses of size $k$ with possible repetitions. Clauses are sets of $k$ independently sampled variables, without repetitions, where each variable $x_i$ is chosen with probability $P(x_i) \sim i^{-\beta}$, and negated with probability $1/2$. 
In this paper, we also study the SAT-UNSAT phase 
transition phenomena in this new model using percolation techniques of 
statistical mechanics. We prove that random scale-free formula over $n$ variables, exponent $\beta$ and $\mathcal{O}(n^{(1 - \beta)k})$ clauses of size $k$ are unsatisfiable with high probability (see Theorem~\ref{thm-exponent}). This means that, for big enough values of $\beta$, the number of clauses needed to make a formula unsatisfiable is sub-linear on the number of variables, contrarily to the standard random SAT model. We also prove that scale-free random 2-SAT formulas with exponent $\beta<1/2$ and a ratio of clause/variables $m/n > \frac{1-2\beta}{(1-\beta)^2}$ are also unsatisfiable with high probability (see Theorem~\ref{thm-scale-free-2-SAT}). This last result, together with a coincident lower bound found by Friedrich et al.~\cite{tobiasAAAI17} allow us to conclude that scale-free random 2SAT formulas show a SAT-UNSAT phase transition threshold.

During the revision of this article, many new results related to the phase transition on scale-free random formulas have been found. \citet{tobiasESA17} generalize the notion of scale-free random k-SAT formulas and prove that there exists an asymptotic satisfiability threshold (in the sense of \cite{friedgut}) for $\beta<1-1/k$, when the number of clauses is linear in the number of variables. \citet{tobiasSAT18} find sufficient conditions for the sharpness of this threshold, generalizing \citet{friedgut}'s result for uniform random formulas. \cite{tobiasICALP19} generalize the notion of scale-free formula to the notion of \emph{non-uniform random formula}, only assuming that variable $x_i$ is selected with probability $p_i$ (where $p_1\geq p_2\geq\dots\geq p_n$) and determine the position of the threshold for $k=2$. \citet{cooper} and \citet{bulatov} analyze the \emph{configuration model} for 2-SAT where, instead of fixing the probability of every variable, they fix the degree of every variable. If these degrees follow a power-law distribution, the location of the satisfiability threshold (for $k=2$) is the same as in our model.

This article proceeds as follows. In Section~\ref{sec-graphs} we
review some methods to generate scale-free random graphs. One of these
methods is the basis of the definition of scale-free random formulas,
introduced in Section~\ref{sec-formulas}. In
Section~\ref{sec-industrial}, we summarize some properties of industrial
or real-world SAT instances, described in detail in our work presented
at CP'09~\cite{CP09}. We prove the existence of a SAT-UNSAT phase
transition phenomenon in scale-free random 2-SAT instances in
Section~\ref{sec-2SAT}. This is done using 
percolation techniques. In Section~\ref{sec-small-cores}, we prove that
when the $\beta$ parameter that regulates the scale-free struture of
formulas exceeds a certain value, the SAT-UNSAT phase transition
phenomena vanishes, and most formulas become unsatisfiable due to
small cores of unsatisfiable clauses.

\section{Generation of Scale-Free Graphs}\label{sec-graphs}

Generating scale-free formulas has an obvious relationship with the generation of scale-free graphs. In this section we review some graph generation methods developed by researchers on complex networks.

A scale-free graph is a graph where node degrees follow a power-law distribution $P(k) \sim k^{-\gamma}$, at least asymptotically, where exponent $\gamma$ is around $3$. Preferential attachment~\cite{Barabasi} has been proposed as the natural process that makes scale-free networks so prevalent in
nature. This process can be used to generate scale-free graphs as
follows. Given two numbers $n$ and $m$, we start at time $t=m+1$ with
a clique of size $m+1$ where all nodes have degree $m$ (in the limit
when $n$ tends to infinity, the starting graph is not relevant).
Then, at every time $t=m+2,\dots, n$, we add a new node (with index
$t$), connected to $m$ distinct and older nodes $s<t$, such that the
probability that a node $s$ gets a connection to this new node $t$ is
proportional to the degree of $s$ at time $t$. This process generates
a scale-free graph with asymptotic exponent $\gamma=3$, average node
degree $\E[k]=2\,m$ and minimum degree $k_i \geq m$, for all nodes. We
can also prove that the expected degree of node $i$ is $\E[k_i]\sim
i^{-\beta}$, for small values of $i$, where $\beta=0.5$.

In order to explain the origin of scale-free networks where $\gamma
\neq 3$, several
models have been proposed~\cite{evolution-networks}. One of these models is based on the
\emph{aging} of nodes~\cite{aging}. This means that the probability of
a node $s$ (created at instant $s$) to get a new edge at instant $t$ is proportional to the product of
its degree and $(t-s)^{-\alpha}$, where $t-s$ is the \emph{age}
of the node. This model generates scale-free graphs
when $\alpha<1$. When $\alpha\to 0$, the exponents of the power-laws
$P(k)\sim k^{-\gamma}$ and $\E[k_i]\sim i^{-\beta}$ are $\gamma = 3 +
4(1-\ln 2)\alpha$ and $\beta = 1/2 -(1-\ln 2)\alpha$,
respectively. Therefore, the value of $\alpha$ may be used to tune the
values of $\gamma$ and $\beta$.

In the previous methods, \emph{growth} in the number of nodes is essential.  There are other methods, usually called \emph{static}, where the number of nodes is fixed from the beginning and during the process we only add edges.

The simplest method, assuming uniform probability for all graphs with a scale-free degree distribution in the degree of nodes, is the \emph{configuration method}, that can be implemented as follows.

Given a desired number of nodes $n$ and exponent $\gamma$, for every node $i\in\{1,\dots,n\}$, generate a degree $k_i$ following the probability $P(k)=k^{-\gamma}/\zeta(\gamma)$, independently of $i$.  Here, $\zeta(x)=\sum_{i=1}^\infty i^{-x}$ is the Riemann zeta function. Then, generate a graph with these node degrees, ensuring that all them are generated with the same probability. This can be done, for instance, with a unfold-fold process: In the unfolding, we replicate node $i$, with degree $k_i$, into $k_i$ new nodes with degree $1$.  Then, we randomly generate a graph where all nodes have degree equal to one, ensuring that all 1-regular graphs with $\sum_{i=1}^n k_i$ nodes are generated with the same probability. Then, in the folding, we merge the $k_i$ nodes that came from the replication of $i$, into the same node. When there is an edge between two nodes and these two nodes are merged, a self-loop is created. Similarly, when we have two edges $i_1\leftrightarrow j$ and $i_2\leftrightarrow j$ and $i_1$ and $i_2$ are merged, a duplicated edge is created. Therefore, we reject the resulting graph, if it contains self-loops or multiple edges between the same pair of nodes.  Alternatively, we can also apply the Erd\"os-R\'enyi generation method to the unfolded set of nodes, with average node degree equal to one. In this later case, we would ensure that after folding, node $i$ has a degree close to $k_i$, since in the Erd\"os-R\'enyi model, node degrees follow a binomial distribution (a Poisson distribution $P(k)=e^{-z}\frac{z^k}{k!}$ in the infinite limit, where $z$ is the average degree, $z=1$ in our case).

The previous method has two problems. First, the resulting graph (after the unfold-fold process) will have average node degree equals to:
\[
\E[k] = \sum_{k=1}^\infty k\,P(k) =\sum_{k=1}^\infty k\, \frac{k^{-\gamma}}{\zeta(\gamma)}
= \frac{\zeta(\gamma-1)}{\zeta(\gamma)}
\]

If we want to obtain a graph with a distinct average degree, we have to modify the probability $P(k)$ for small values of $k$, and ensure that $P(k)$ follows a power-law distribution only asymptotically for big values of $k$. In other words, we only require $P(k)$ to follow a \emph{heavy-tail} distribution. Second, a great fraction of generated graphs will contain self-loops or multiple-edges after folding. This means that a great fraction of graphs will be rejected, which makes the method inefficient. However, the model can be useful to translate some properties of the Erd\"os-R\'enyi model to scale-free graphs via the unfolding-folding process and the \emph{configuration model}~\cite{Bender&Canfield,Bollobas}.

The unfolding-folding procedure was described by~\citet{regularpowerlaw}.  They, instead of assigning a random degree to each node, describe a model where, given two parameters $\alpha$ and $\gamma$,\footnote{In the original paper, authors use the name $\beta$ instead of $\gamma$.} we choose a random graph (with uniform probability, and allowing self-cycles) among all graphs satisfying that the number of nodes with degree $x$ is $e^\alpha/x^\gamma$. When $\gamma>2$, the average node degree in this model is also $\frac{\zeta(\gamma-1)}{\zeta(\gamma)}$.

Alternatively, instead of fixing the degree of every node, we can fix the \emph{expected degree} of every node $\E[k_i]= w_i$. In order to construct a graph where nodes have this expected degree $\E[k_i]\sim w_i$, we only need to generate edge $i\leftrightarrow j$ with probability $P(i\leftrightarrow j) \sim w_i\, w_j$. If we want to generate a scale-free graph where $P(k) \sim k^{-\delta}$, for sparse graphs, it suffices to fix $w_i= i^{1/(\delta-1)}$~\cite{model-chung-lu,universal-load-distribution} (see also Theorem~\ref{thm-exponents-relation}).

Our scale-free formula generation method is based on this static scale-free graph generation model with fixed expected node degrees. Basically, nodes are replaced by variables. Then, instead of edges, we generate hyper-edges. Negating every variable connected by a hyper-edge with probability $1/2$, we get clauses.

\section{Scale-Free Random Formulas}\label{sec-formulas}

In this section we describe the \emph{scale-free random SAT formulas} model.

We consider $k$-SAT formulas over $n$ variables, denoted by $x_1,\dots,x_n$. A formula is a conjunction of $m$ possibly repeated clauses, represented as a multiset. Clauses are disjunctions of $k$ literals, noted $l_1\vee\dots\vee l_k$, where every literal may be a variable $x_i$ or its negation $\neg x_i$. We identify $\neg\neg x$ with $x$. We restrict clauses to not contain repeated occurrences of variables. This avoids simplifiable formulas like $x\vee x\vee y$ and tautologies like $x\vee \neg x\vee y$.
In general, we represent every variable by its index, and negation as a minus, writing $i$ instead of $x_i$, and $-i$ instead of $\neg x_i$. In other words, a variable $x$ is a number in $\{1,\dots,n\}$, and a literal a number in $\{-n,\dots,n\}$ distinct from zero. 
We use the notation $\pm x$ to denote either $x$ or $\neg x$. 
The number of occurrences of literal $l$ in a formula is denoted by $k_l$, and $K_x = k_x+k_{\neg x}$ denotes the number of occurrences of variable $x$.  
The size of a formula $F$ is $|F|=m\,k$.

In the following, we will use the notation $P(x)\sim f(x)$ to indicate that random variable $x$ follows the probability distribution $f(x)$. The notation $f(n)\approx g(n)$ indicates that $\lim_{n\to \infty} \frac{f(x)}{g(x)} = 1$.\footnote{Notice that $f(n)\approx g(n)$ is equivalent to $f(n)=g(n)(1+o(1))$. When $g(n)=\Theta(1)$, then $f(n)=g(n)(1+o(1))$ is equivalent to $f(n)=g(n)+o(1)$.}

\begin{definition}[Scale-free Random Formula]
In the scale-free model, given $n$, $m$ and $\beta$, to construct a random formula, we generate
$m$ clauses independently at random from the set of $2^k\,{n\choose
  k}$ clauses, sampling every valid clause with probability
\[
P(l_1\vee\dots\vee l_k) \sim \prod_{i=1}^k P(l_i)
\]
where every literal $l_i$ is sampled with probability
\[
P(x) = P(\neg x) \sim x^{-\beta}
\]
\end{definition}

In practice, we generate a variable $x$ with probability $P(x)=x^{-\beta}/\sum_{i=1}^n i^{-\beta}$, negate it with probability $1/2$, repeat the process $k$ times, and reject clauses containing repeated variables.

Therefore, the probability of a clause satisfies the inequality
\[
P(l_1\vee\dots\vee l_k) \geq
\frac
    {k!\,\prod_{i=1}^k |l_i|^{-\beta}}
    {(2\,\sum_{i=1}^ni^{-\beta})^k}
    \]
    
\subsection{Some Properties of the Model}

In the case of the graph generator, we reject self-loops and repeated edges between two nodes. This makes distribution of degrees to follow a power-law, only asymptotically and for sparse graphs.  In our case, we reject clauses with repeated variables. This is the reason that invalidates the reverse direction in the previous inequality. It also makes formulas to follow a power-law distribution in the number of variable occurrences only asymptotically (see Theorem~\ref{thm-exponents-relation}). In the following we will discuss when the approximation for $P(l_1\vee\dots\vee l_k)\approx\frac{k!\,\prod_{i=1}^k |l_i|^{-\beta}}{(2\,\sum_{i=1}^ni^{-\beta})^k}$ is tight (Lemma~\ref{lem-surname}).

Notice that $\sum_{i=1}^n i^{-\beta} = H_{n,\beta}$ are the \emph{generalized harmonic numbers}. When $n$ tends to infinity and $\beta\neq 1$, using the Euler-Maclaurin formula, they can be approximated as
\begin{equation}
\sum_{i=1}^n i^{-\beta} = \zeta(\beta) +\frac{1}{1-\beta}n^{1-\beta} + \frac{1}{2}n^{-\beta}+\mathcal{O}(n^{-\beta-1})
\label{eq-euler}
\end{equation}
where $\zeta(\beta)$ is the Riemann zeta function.
When $\beta=1$, we have 
\begin{equation}
\sum_{i=1}^n i^{-1}=\gamma+\log n+\mathcal{O}(n^{-1})
\label{eq-euler2}
\end{equation}
where $\gamma$ is the Euler constant. 

This means that, when $n$ tends to infinity, the probability of sampling variable $x_i$ is $P(x_i)=o(1)$, when $0\leq \beta\leq 1$, and $P(x_i) = i^{-\beta}/\zeta(\beta)+o(1)$, when $\beta>1$. The fact that the probability of sampling a variable does not vanish, when the number of variables tend to infinity and $\beta>1$, may be troublesome. In particular, the probability of generating clauses with duplicated variables does not vanish, even for constant clause sizes. Similarly, to avoid duplicated variables, we also have to impose an upper bound $k=o(n^{\min\{1/2,1-\beta\}})$.

\begin{lemma}\label{lem-surname}
When $0\leq \beta<1$, the sizes of clauses are $k=o(n^{\min\{1/2,1-\beta\}})$ and $n$ tends to infinity, the probability of generating a clause with a duplicated variable tends to zero.

In these conditions, the probability of a random variable and the probability of a random clause in a formula are
$$
P(x=x_i) \approx \frac{i^{-\beta}}{\sum_{j=1}^n j^{-\beta}}
$$
$$
P(C=l_1\vee\dots\vee l_k) \approx
\frac
    {k!\,\prod_{i=1}^k |l_i|^{-\beta}}
    {(2\,\sum_{i=1}^ni^{-\beta})^k}
$$
\end{lemma}
\begin{proof}
We will use a result known as \emph{surname problem}~\cite{surnameParadox}, that generalizes the \emph{birthday paradox}. Let $X_1,\dots,X_k$ be independent random variables which have an identical discrete distribution $P(X=i)= p_i$, for $i\geq 1$. Let $R_k$ be the coincidence probability that at least two $X_j$ have the same value. Let $r_k=1-R_k$ be the non-coincidence probability. Then, $r_k$ may be computed using the recurrence $r_0=1$ and
$$
r_k = \sum_{j=1}^k (-1)^{j-1} \frac{(k-1)!}{(k-j)!}\,P_j\,r_{k-j}
$$
where $P_k=\sum_{i\geq 1} (p_i)^k$.
The coincidence probability can be computed as $R_1=0$ and
$$
R_k = R_{k-1} + \sum_{j=2}^k (-1)^j\frac{(k-1)!}{(k-j)!}\,P_j(1-R_{k-j})
$$

In our case, we face the problem of choosing $k$ independent variables, and we want to compute the probability of getting a duplicated variable, hence a rejected clause. When $\beta<1$, we have:
$$
P_k = \frac{\sum_{i\geq 1}i^{-\beta k}}{(\sum_{i\geq1}i^{-\beta})^k}=
\frac{\frac{1}{1-\beta k}n^{1-\beta k}+\zeta(\beta k)+\mathcal{O}(n^{-\beta k})}
{(\frac{1}{1-\beta}n^{1-\beta}+ \mathcal{O}(1))^k} = 
\frac{(1-\beta)^k}{1-\beta k}n^{1-k} +\zeta(\beta k)(1-\beta)^k n^{-(1-\beta)k} + \mathcal{O}(n^{-k}) 
$$
Depending on whether $\beta\, k$ is greater or smaller than $1$, the first or the second term of $P_k$ will dominate.

Since $\frac{(k-1)!}{(k-j)!} < k^{j-1}$ and $R_{k-j}\geq 0$, we have
$$
R_k \leq \sum_{i=2}^k\sum_{j=2}^i i^{j-1}\,P_j\leq k\max_{j=2,\dots,k}k^{j-1}P_j
$$
In our case, assuming $k=\mathcal{O}(n^\alpha)$, and replacing the value of $P_j$, we get
$$
R_k \leq \mathcal{O}(n^\alpha) \max_{j=2,\dots k}\mathcal{O}(n^{\alpha(j-1)} (n^{1-j}+n^{-(1-\beta)j})) = 
\max_{j=2,\dots k}\mathcal{O}(n^{1-(1-\alpha)j} +n^{-(1-\beta-\alpha)j}))
$$
Assuming $\alpha \leq 1-\beta < 1$, the maximum is obtained for $j=2$. In this situation
$
R_k \leq \mathcal{O}(n^{-(1-2\alpha)}+n^{-1(1-\beta-\alpha)})
$.
Therefore, it suffices to assume that $\alpha <\min\{1/2,1-\beta\}$ to ensure that $R_k=o(1)$.
\end{proof}


\begin{lemma}\label{lem-expectedK}
In a scale-free random formula over $n$ variables and $m=C\,n$ clauses of size $k=\mathcal{O}(1)$, generated with exponent $0<\beta<1$, the expected number of occurrences of variable $x_i$ is
\[
\E[K_i] \approx C\,k\,(1-\beta)\left(\frac i n\right)^{-\beta}
\]
\end{lemma}
\begin{proof}
By Lemma~\ref{lem-surname} and equation~(\ref{eq-euler}), since $0<\beta<1$ we have
\[
\E[K_i]= P(i)\,|F| \approx \frac{i^{-\beta}}{\zeta(\beta)+\frac{1}{1-\beta}n^{1-\beta}+\mathcal{O}(n^{-\beta})}\,C\,k\,n \approx C\,k\,(1-\beta)\left(\frac i n\right)^{-\beta}
\]
\end{proof}

The following theorem ensures that the formulas we get are scale-free,
in the sense that the number of occurrences of variables follow a
power-law distribution $P(K)\sim K^{-\delta}$, for big enough values of $K$.

\begin{theorem}\label{thm-exponents-relation}
  In scale-free random formulas over $n$ variables, with $m=C\,n$ clauses of size $k$, and generated with exponent $0<\beta<1$, when $n$ tends to $\infty$ being $C$ and $k$ constants, the probability that a variable has $K$ occurrences, where $K=\Omega(\sqrt{n\log n})$ or $K=\Omega\left( (n^2\log n)^{\frac\beta{2+\beta}}\right)$, follows a power-law distribution $P(K) \sim K^{-\delta}$, where $\delta = 1/\beta + 1$.
\end{theorem}

\begin{proof}
In the limit when $n\to\infty$, by Lemma~\ref{lem-surname}, $P(x_i) \approx \mathfrak{C}\, i^{-\beta}$ is the probability of sampling a variable $x_i$, for some constant $\mathfrak{C}=1/\sum_{j=1}^n j^{-\beta}\approx (1-\beta)\,n^{\beta-1}$ that depends on $n$.  
Let $K_i$ be the number of occurrences of variable $i$ in a randomly generated formula $F$.  
We have $\E[K_i] = |F|\,\mathfrak{C}\,i^{-\beta}$.  
Chernoff's or Hoeffding's bounds ensure that, under certain conditions that we will consider later, $K_i$ is approximately $\E[K_i]$.  
Hence, $K_1 > K_2 > \cdots > K_n$ with high probability.

Now we want to approximate the probability $F(K) = \int_{k=K}^\infty P(k)\,\mathrm{d}k$ that a variable occurs at least $K$ times.  Given a value $K$, let $i$ be the index of the variable satisfying $\E[K_i] = K$. 
Under these conditions, all variables with index smaller that $i$ will have more than $K$ occurrences, and those with indexes between $i+1$ and $n$ have less than $K$ occurrences.  
Therefore, $F(K) = i/n$, for the particular $i$ defined above. 
From $\E[K_i] = K$ and $\E[K_i] = |F|\,C\,i^{-\beta}$ we obtain
\[ 
F(K) = \frac{i}{n} = \frac{1}{n}\,\left(\frac{K}{|F|\,\mathfrak{C}}\right)^{-1/\beta}
\]
Then, the probability $P(K)$ is
\[
P(K) = -\frac{\partial}{\partial K} F(K) =
\frac{(|F|\,\mathfrak{C})^{1/\beta}}{\beta\,n}\, K^{-1/\beta -1}
\]
Hence we obtain a discrete power-law distribution with exponent $\delta = 1/\beta +1$.

The problem is that $\E[K_i]$ is a good approximation of $K_i$ only when $i$ is small. For instance, when $i=\Omega(n)$, we have $P(x_i)=\Theta(n^{-1})$ and $\E[x_i] = \Theta(1)$. In this situation, when $n\to\infty$ being $C$ and $k$ constants, the number of occurrences $K_i$ of the variable $x_i$ follows a Poisson distribution with constant variance. This means that, even in the limit $n\to\infty$, we can not assume that $i<j$ implies $K_i > K_j$, when $i=\Omega(n)$.
In the following we will find an upper bound for the index $i$ of the variable (a lower bound for the value of $K$) ensuring that $\E[K_i]$ is a good approximation of $K_i$, when $n\to\infty$. We will use both Hoeffding's and Chernoff's bounds.

In what follows, let be $C$ be the constant such that $|F|\approx C\,n$ is the size of the formula.

Hoeffding's bound states that, if $X=X_1+\dots+X_n$ is the sum of identical and independent Bernoulli variables, then
\[
P(|X-\E[X]|\geq \epsilon\, n) \leq 2\,e^{-2\epsilon^2n}
\]
taking $\epsilon=\sqrt{\frac{\log n}n}$ we obtain
\[
P\left(|X-\E[X]|\geq \sqrt{n\log n}\right) \leq \frac 2{n^2}
\]

Given a value of $K$, let's fix two variables $i$ and $j$ such that
\[
\begin{array}{l}
\E[K_i] = K \\
\E[K_j] = K - \sqrt{n\log n}
\end{array}
\]

We have $P(K_j \geq K) \leq 2/n^2$, and for all variables $r$ with bigger indexes
\[
\sum_{r\geq j} P(K_r\geq K) = o(1)
\]
We have already argued that $F(K)=P(k\geq K) \approx i/n$. Using $j$, we have a strict bound
\[
P(k\geq K) \leq j/n + o(1)
\]

By Lemma~\ref{lem-expectedK}, we get
\[
\begin{array}{rl}
K =  &
\E[K_i] \approx C(1-\beta)\left(i/n\right)^{-\beta} \\
K-\sqrt{n \log n} = &
\E[K_j] \approx C(1-\beta)\left( j/n\right)^{-\beta} =
C(1-\beta)\left(j/i\right)^{-\beta}\left(i/n\right)^{-\beta} \approx 
K\left(j/i\right)^{-\beta}
\end{array}
\]
Therefore 
\[
\begin{array}{l}
j \approx i\left(1-\frac{\sqrt{n\log n}}{K}\right)^{1/\beta} \\
i \approx \left(\frac{C(1-\beta)}{K}\right)^{1/\beta}
\end{array}
\]
Replacing the expressions for $i$ and $j$, we get
\[
P(k\geq K) \leq j/n + o(1) 
\approx 
\left(\frac{C(1-\beta)}K\left(1-\frac{\sqrt{n\log n}}K\right)\right)^{1/\beta}
\]
If $K=\Omega(\sqrt{n\log n})$ then
\[
P(k\geq K) \leq \left(\frac K{C(1-\beta)}\right)^{-1/\beta}+o(1)
\]
Similarly, we can prove the same lower bound $P(k\geq K)\geq \left(\frac K{C(1-\beta)}\right)^{-1/\beta}+o(1)$, using now the variable $j$ such that $\E[K_k]=K+\sqrt{n\log n}$.

Alternatively, we can use the Chernoff's bound
\[
P(|X-\E[X]|\leq \delta \E[X]) \leq 2\,e^{-\frac{\delta^2\E[X]}3}
\]
where $X$ is the sum of independent random variables in the range $[0,1]$.
In order to ensure that the $K_i$'s are sorted, we require that, in the limit $n\to \infty$, we have $P(K_{i}<K_{i+1})=\mathcal{O}(n^{-1})$.
We take the value of $\delta$ that satisfies
\[
\delta \E[K_i] = \frac{\E[K_i]-\E[K_{i+1}]}2
\]
By Lemma~\ref{lem-expectedK} and the Taylor expansion $(1+x)^{a}=1+a\,x+\mathcal{O}(x^2)$ this value of $\delta$, when $i\to\infty$, is
\[
\delta \approx 1/2 - 1/2\left(\frac{i+1}i\right)^{-\beta} 
\approx 
\frac{\beta}{2i}
\]
And, for this value of $\delta$, we impose
\[
2\,e^{-\frac{\delta^2\E[K_i]}3} 
\approx 
2\exp\left({-\frac{(\beta/2i)^2C(1-\beta)(i/n)^{-\beta}}3}\right)=\mathcal {O}(n^{-1})
\]
From this, we get the minimum value of $i$ for which $P(K_{i}<K_{i+1})=\mathcal{O}(n^{-1})$. 
\[
i = \mathcal{O}\left(n^{\beta/(2+\beta)}/\log^{1/(2+\beta)} n\right)
\]

The value of $K=\E[K_i]$ corresponding to this variable $x_i$ gives us a value from which on we can expect to observe the power-law distribution in $P(K)$.
\[
K = \Omega\left((n^2\log n)^{\beta/(2+\beta)}\right)
\]
\end{proof}

\subsection{Implementation of the Generator}

The generation method is formalized in Algorithm~\ref{fig-alg}.
\begin{algorithm}
\KwIn{$n, m, k, \beta$} 
\KwOut{ a $k$-SAT instance with $n$ variables and $m$ clauses} 
$F = \emptyset$\;
\ForEach{$i=1,\dots,m$}{
  \Repeat{$C_i$ does not contain repeated variables}{
    $C_i = \eclause$\;
    \ForEach{$j=1,\dots,k$}{
      v = sampleVariable($\beta$,n)\;
      $C_i = C_i \vee (-1)^{rand(2)}\cdot v$;\\
    }
  }
  $F = F\cup \{C_i\}$\\
}
\caption{Scale-free random $k$-SAT formula generator.}\label{fig-alg}
\end{algorithm}
The function sampleVariable($\beta$,n) may be implemented in two ways.

We can compute a vector $p$ such that $p[i]= \sum_{j=1}^i j^{-\beta}/\sum_{j=1}^n j^{-\beta}$ at the beginning of the algorithm. Then, every time we call sampleVariable, we compute a random number $r$ uniformly distributed in $[0,\,1)$, using a dichotomic search, look for the smallest $i$ such that $p[i] > r$, and return such $i$.

Alternatively, if $n$ is big we can use the following approximated algorithm. If we want to generate numbers $x$ with probability density $f(x)$, we can integrate $F(x)=\int f(x)\,\mathrm{d}x$, find the inverse function, and compute $F^{-1}(y)$, where $y$ is a uniformly random number in $[0,\, 1]$. Our probability function is discrete. However, when $0<\beta<1$, and both $X\to \infty$ and $n\to\infty$, we can approximate it as
\[
P(x\leq X) = \frac{\sum_{i=1}^X i^{-\beta}}{\sum_{i=1}^n i^{-\beta}} \approx
\frac{\zeta(\beta)+1/(1-\beta)\,X^{1-\beta}}{\zeta(\beta)+1/(1-\beta)\,n^{1-\beta}}
\]
Therefore, computing the inverse, sampleVariable may be computed as
\[
X = \left\lfloor \Big(\big(n^{1-\beta} +(1-\beta)\zeta(\beta)\big)\, Y-(1-\beta)\,\zeta(\beta)\Big)^{1/(1-\beta)}\right\rfloor+1
\]
where $Y$ is a uniform random variable in $[0,\, 1)$. This way, avoiding the use of the vector $p$ and the dichotomic search, we save a $\mathcal{O}(\log n)$ factor in the time-complexity and a $\mathcal{O}(n)$ factor in the space-complexity of the generator.

\section{Industrial SAT Instances}\label{sec-industrial}

\begin{figure}
  \hbox to \columnwidth {
    \hss
    \includegraphics[width=0.55\columnwidth]{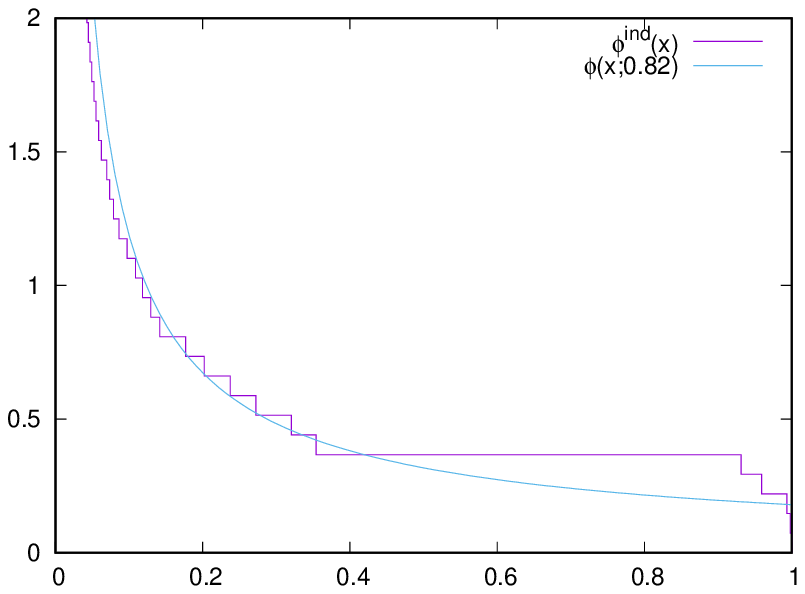}
    \hss
    \includegraphics[width=0.55\columnwidth]{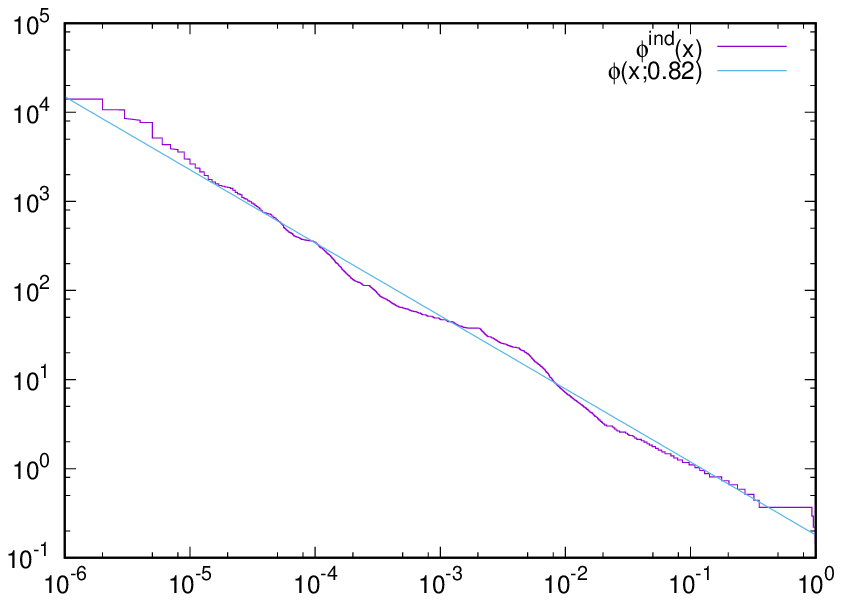}
    \hss}
\caption{Estimated industrial function $\phi^{ind}(x)$ (in red) and power-law function $\phi(x;0.82)=(1-0.82)\,x^{-0.82}$ (in blue), with normal axes (left) and double-logarithmic axes (right).}
\label{fig-ind}
\end{figure}

\begin{figure}
  \hbox to \columnwidth{
    \hss
    \includegraphics[width=0.5\columnwidth]{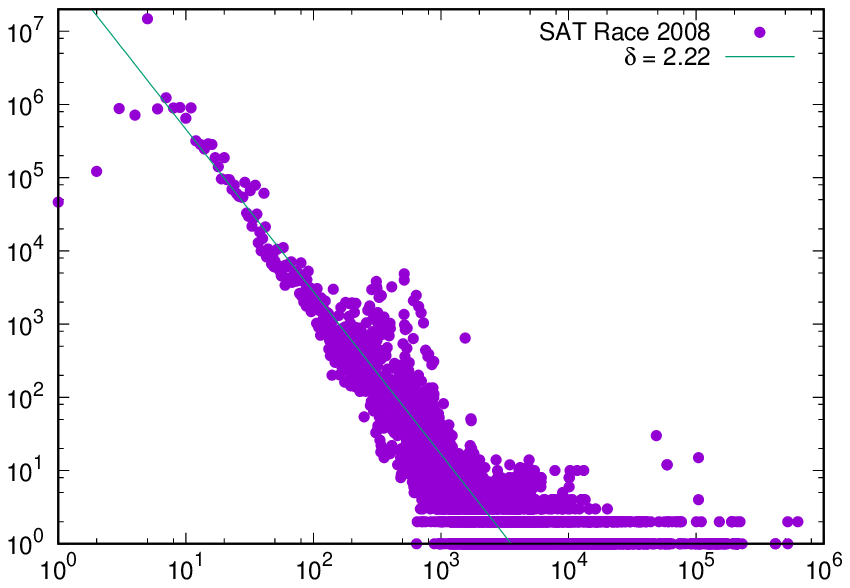}
    \hss
    \includegraphics[width=0.5\columnwidth]{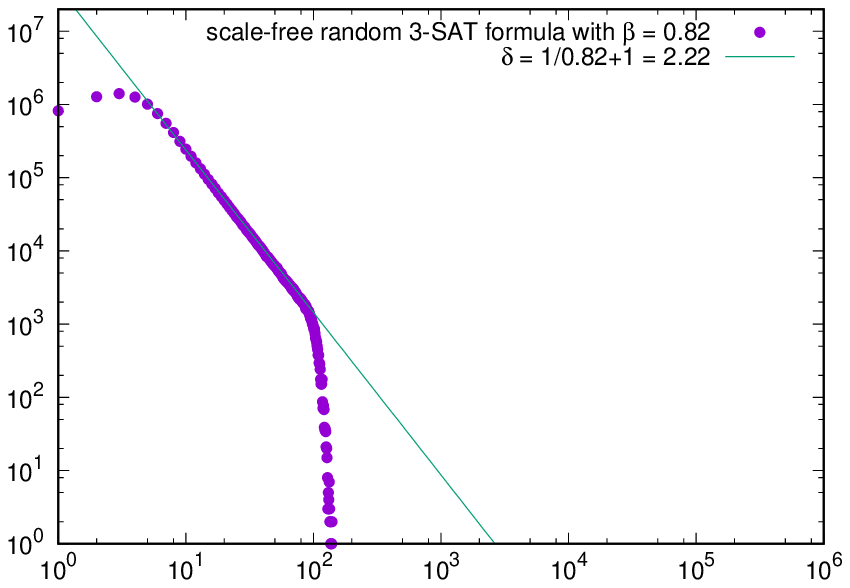}
    \hss
}
\caption{Comparison of the frequencies of variable occurrences obtained for the whole set of instances used in the SAT Race 2008, and for a scale-free random 3-SAT formula generated with $\beta=0.82$, $n=10^7$ and $m=2.5\cdot 10^7$. In both cases, the $x$-axis represents the number of occurrences, and the $y$-axis the number of variables with this number of occurrences. Both axes are logarithmic. It is also shown the line with slope $\alpha = 1/0.82 + 1 = 2.22$, corresponding to the function $f(x)= C\,x^{-2.22}$ in double-logarithmic axes.}
\label{fig-experiment}
\end{figure}

In the previous section we have scale-free random SAT instances. We want this models to generate formulas as close as possible to industrial ones.  Therefore, we want to compute the value of $\beta$ that best fits industrial instances. For this purpose we have studied the 100 benchmarks (all industrial) used in the SAT Race 2008. All together, they contain $n = 25693792$ variables, with a total of $\sum_{i=1}^{n} K_i = 349760681$ occurrences. Therefore, the average number of occurrences per variable is $E\left[K_i\right] = \sum_{i=1}^{n} K_i /n = 13.6$. If we used the classical (uniform) random model to generate instances with this average number of occurrences, most of the variables would have a number of occurrences very close to $13.6$.  However, in the analyzed industrial instances, close to $90\%$ of the variables have less than this number of occurrences, and more than $60\%$ have $6$ or less occurrences.  The big value of the average is produced by a small fraction of the variables that have a huge number of occurrences. This indicates that the number of occurrences could be better modeled with a power-law distribution. This was already suggested by \citet{regularkSAT}.

In order to check if those industrial instances (all together) are scale-free SAT formulas, and estimate the value of $\beta$, we compute the number of occurrences of each variable of each industrial instance. Then, we rename the indexes of such variables such that $K_i \geq K_{i+1}$, for $i=1,\dots,n-1$.  Now, before comparing $K_i$ with $i^{-\beta}/\sum_{j=1}^n j^{-\beta}$, we renormalize both functions such that both are defined in $[0,\, 1]$ and its integral in this range is $1$.  Hence, we define for the empirical $K_i$, the \emph{empirical} function $\phi^{ind}$ as
\[
\phi^{ind}(x) =_{def} \frac{n}{\sum_{j=1}^{n} K_j}\,K_{\lfloor n\,x\rfloor}
\]
and, for the theoretical function $P(i)$, the \emph{theoretical} function $\phi(x;\beta,n)$ as
\[
\phi(x;\beta,n) =
\frac{n}{\sum_{j=1}^{n} j^{-\beta}}\,\lfloor n\,x\rfloor^{-\beta} \approx 
\frac{n}{\zeta(\beta)+\frac 1{1-\beta}n^{1-\beta}}(n\, x)^{-\beta} =
\frac{1-\beta}{\frac{(1-\beta)\zeta(\beta)}{n^{1-\beta}}+1}x^{-\beta}
\]

When $\beta<1$, we get
\[
\phi(x;\beta)=\lim_{n\to\infty}\phi(x;\beta,n) = (1-\beta)x^{-\beta}
\]

In Figure~\ref{fig-ind} we represent both functions with normal axes, and with double-logarithmic axes. Notice that in double logarithmic-axes, the slope of $\phi^{ind}(x)$ allows us to estimate the value of $\beta=0.82$.

Theorem~\ref{thm-exponents-relation} allows us to ensure that the distribution of frequencies on the number of occurrences of variables follows a power-law distribution, with exponent $\delta = 1 / 0.82 + 1 = 2.22$.

Finally, we have generated a scale-free random 3-SAT formula with $n=10^7$ variables, $m=2.5\cdot 10^7$ clauses and $\beta = 0.82$. In Figure~\ref{fig-experiment}, we show the frequencies of occurrences of variables of this formula and compared it with those obtained for the SAT Race 2008, and the line with slope $\alpha = 1/0.82 + 1= 2.22$.

\section{Phase Transition in Scale-Free Random 2-SAT Formulas}\label{sec-2SAT}

\citet{mickgetssome} proved that a random formula with $(1+o(1))cn$ clauses of size $2$ over $n$ variables, is satisfiable with probability $1-o(1)$, when $c<1$, and unsatisfiable with probability $1-o(1)$, when $c>1$, where $o(1)$ represents a quantity tending to zero as $n$ tends to infinity.

As will see in this section, a similar result for scale-free random 2-SAT formulas can be obtained using percolation and mean field techniques.

Percolation theory describes the behavior of connected components in a graph when we remove edges randomly. Erd\"os and R\'enyi~\cite{erdos-renyi59} are considered the initiators of this theory. In this seminal paper on graph theory they proposed a random graph model $G(n,m)$ where all graphs with $n$ nodes and $m$ edges are selected with the same probability. Gilbert~\cite{gilbert} proposed a similar model $G(n,p)$ where $n$ is also the number of nodes, and every $n \choose 2$ possible edge is selected with probability $p$. For not very sparse graphs (when $p\,n^2\to \infty$), both models have basically the same properties taking $m={n \choose 2}p$. Erd\"os and R\'enyi~\cite{erdos-renyi60} also studied the connectivity on these graphs and proved that

\begin{itemize}
\item when $n\,p<1$, i.e. $m < n/2$, a random graph almost surely has no connected component larger than $\mathcal{O}(\log n)$,
\item when $n\,p=1$, i.e. $m = n/2$ a largest component of size $n^{2/3}$ almost surely emerges, and
\item when $n\,p>1$, i.e. $m > n/2$, the graph almost surely contains a unique giant component with a fraction of the nodes and no other component contains more than $\mathcal{O}(\log n)$ nodes.
\end{itemize}

Phase transitions is a phenomenon that has been observed and studied in many AI problems. Many problems have an order parameter that separates a region of solvable and unsolvable problems, and it has been observed that hard problems occur at critical values of this parameter. Mitchell et al.~\cite{mitchell} found this phenomena in 3-SAT when the ratio between number of clauses and variables is $m/n \approx 4.3$. Gent and Walsh~\cite{gent-walsh} observed the same phenomenon with clauses of mixed length. 

There is a close relationship between SAT problems and graphs. Both, percolation on graphs and phase transition in SAT (or other AI problems) are critical phenomena and both can be studied using mean field techniques from statistical mechanics. Percolation theory has been used and inspired works in the literature about random SAT and satisfiability threshold, e.g. in~\citet{achlioptas} to determine the satisfiability threshold of 1-in-k SAT and NAE 3-SAT formulas. Some results on graphs have been previously extended to 2-SAT. For instance,~\citet{Sinclair13} adapted Achlioptas processes for graphs into formulas. \citet{Bollobas01} investigated the scaling window of the 2-SAT phase transition, finding the critical exponent of the order parameter and proving that the transition is continuous, adapting results of~\citet{Bollobas84} for Erd\"os-R\'enyi graphs. The relationship between percolation in random graphs and phase transition in random 2-SAT formulas is suggested in many other works. For instance, \citet{2+p} when studying the phase transition in $2+P$-SAT (a mixture of $(1-p)m$ clauses of size $2$ and $p\,m$ clauses of size $3$) already mention that \emph{``It is likely that the 2SAT transition results from percolation of these loops...''}. 
\citet{cooper} use the emergence of a giant component in a graph to prove the existence of a phase transition in 2-SAT random formulas with prescribed degrees, using the configuration model. They find, for this model, the same criterion as \citet{tobiasAAAI17} and us in Theorem~\ref{thm-criterion}.
  
Given a random 2-SAT formula with $m$ clauses over $n$ variables, we can construct an Erd\"os-R\'enyi graph where the $2\,n$ literals are nodes, and the $m$ clauses are edges. At the percolation point $m=(2\,n)/2$ of this graph a \emph{giant component} emerges. Just at the same point $m=n$ the 2-SAT phase transition threshold is located.  However, despite the coincidence in the point, the relation between both facts is not direct: a giant component in the graph is not the same as a giant (hence, unsatisfiable) loop of implications in the SAT formula. The connection between two edges $a\leftrightarrow b$ and $b\leftrightarrow c$ in the graph is given by a common node (literal) $b$. Whereas, in the SAT formula, the resolution between $a\vee b$ and $\neg b\vee c$ is through a variable $b$ that is affirmed in one clause and negated in the other. In this section, we elaborate on the relation of giant components
in graphs and unsatisfiability proofs in 2-SAT formulas.

\subsection{A Criterion for Phase Transition in 2-SAT}\label{subsec-criterion}

Unsatisfiability proofs of 2-SAT formulas are characterized by \emph{bicycles}. Let $F$ be a 2-SAT formula. Any sequence of literals $x_1,\dots,x_s$ satisfying $\neg x_i\vee x_{i+1}\in F$, for any $i=1,\dots,s-1$, is called an \emph{implication sequence}. We say that $y$ \emph{implies} $y'$, if there exists an implication sequence of the form $y,x_1,\dots,x_n,y'$. Any implication sequence of the form $x_1,\dots,x_s,x_1$ is called a \emph{cycle}. A \emph{bicycle} is a cycle $x_1,\dots,x_n,x_1$ such that there exists a variable $a$ satisfying $\{a,\neg a\} \subseteq \{x_1,\dots,x_n\}$.

A 2-SAT formula is unsatisfiable if, and only if, it contains a bicycle~\cite{implicationgraph,mickgetssome}.

We will also consider random graphs with $n$ nodes and $m$ edges,\footnote{We will deal with distinct models of random graphs where every graph has a distinct probability of being chosen.} and connected components, defined as subsets of nodes such that any pair of them is connected by a path inside the component. A random graph of size $n$ is said to contain a \emph{giant connected component} if almost surely\footnote{\emph{Almost surely} means that, in the model of random graph, as $n\to\infty$, the probability tends to one.} it contains a connected component with a positive fraction of the nodes. Given a model of random graphs, we say that $c$ is the percolation threshold if any random graph with $n$ nodes and more than $c\,n$ edges almost surely contains a giant component. In a random graph, the degree of a node $x$, noted $k_x$, is a random variable. The random variable $k$ represents the degree of a random node chosen with uniform probability.\footnote{In some of the models of random graphs that we will consider, not all degrees of nodes follow the same probability distribution. Therefore, we will distinguish between $k$ and $k_x$. }

As we commented above, we can represent any 2-SAT formula as a graph where nodes are literals, and clauses $a\vee b$ are edges between literals $a$ and $b$. In classical 2-SAT random formulas, since literals are chosen independently with uniform probability, the generated graph will be an Erd\"os-R\'enyi graph following the model $G(2\,n,m)$. However, a connected component in the graph is not necessarily an unsatisfiability proof of the formula.

First, in a random SAT formula, we may have repeated clauses, which means that from $m$ clauses we will obtain less than $m$ edges. However, for a linear number of clauses, when $\beta<1/2$, there are $(1-o(1))\,m$ distinct clauses or edges. In the classical case, in the limit $n\to\infty$, with a linear number of clauses $m=\mathcal{O}(n)$, and a quadratic number of possible clauses, the probability of any clause is $\mathcal{O}(n^{-2})$, and the probability of being repeated $m\,\mathcal{O}(n^{-2})=\mathcal{O}(n^{-1})$. Therefore, the fraction of repeated clauses is negligible. For scale-free 2-CNF formulas, in Theorem~\ref{thm-exponent}, we will see that, if $\beta<1/2$ then clauses have probability $o(n^{-1})$. Precisely, the most probable 2-CNF clause is $x_1\vee x_2$.
This means, that after generating $\mathcal{O}(n)$ clauses, the probability that a new generated clause has already been generate previously is bounded by
$$
P(x_1\vee x_2)\, \mathcal{O}(n)\approx \frac{2!\,1^{-\beta}2^{-\beta}}{(2\sum_{i=1}^n i^{-\beta})^2}\, \mathcal{O}(n) =
\frac{2^{-1-\beta}}{(\zeta(\beta)+\frac{n^{1-\beta}}{1-\beta}+\mathcal{O}(n^{-\beta}))^2}\mathcal{O}(n) = 
\mathcal{O}(n^{2\beta-1})
$$
This probability bounds the value of the fraction of repeated clauses, that it is meaningless when $\beta<1/2$.

Second, graph connected components and cycles are not the same structure. Therefore, the existence of a giant connected component and the existence of a giant cycle are independent facts.

Molloy and Reed~\cite{Molloy-Reed} and Cohen et al.~\cite{resilence-internet} have studied the existence of giant components in random graphs with heterogeneous and fixed node degrees. Molloy and Reed~\cite{Molloy-Reed} prove that the critical point is at 
$$
Q(\lambda)=\sum_{i>0}i(i-2)\lambda_i=0
$$
where $\lambda_i$ is the fraction of nodes with degree $i$. Whereas, Cohen et al.~\cite{resilence-internet} independently prove (but in a much more informal way) that the critical point is characterized by
$$
\frac{\E[k^2]}{\E[k]}=2
$$
where $k$ is the degree of a random node, and $E$ denotes expectation.
It is easy to see that both criterion are exactly the same.
Interestingly, the criterion depends not only on the expected degree of nodes, but also on the expected square degree of the nodes, hence on the variability of node's degrees. The variability on the nodes degree plays an important role in the location of the percolation threshold. For instance, in the Erd\"os-R\'enyi model, the percolation threshold is located at $m/n =1/2$, hence the expected degree of nodes is $1/2$. However, the expected degree of nodes belonging to the same connected component\footnote{Recall that minimally connected components are trees, where the number of edges is equal to the number of nodes minus one.} of size $r$ is, at least, $(r-1)/r\approx 1$. This discrepancy is only possible if the variability in node's degree is high. This also explains why, in regular random formulas, where we impose variables to occur exactly the same number of times (instead of the same average number of times), we get distinct phase transition thresholds.

Cohen et al.~\cite{resilence-internet} starts assuming that loops of connected nodes may be neglected. In this situation, the percolation transition takes place when a node $i$, connected to a node $j$ in the connected component, is also connected in average to at least one other node, i.e. when $\E[k_i\mid i\leftrightarrow j] = \sum_{k_i} k_i P(k_i\mid i\leftrightarrow j) = 2$. 

Molloy and Reed~\cite{Molloy-Reed} give a more detailed proof that we will try to summarize. Given the list of fixed degrees $k_i$ of every node, they describe a random algorithm that constructs (exposes) all graphs compatible with these degrees with the same probability, exposing connected components one by one:

Let $c_i$ be the degree of node $i$ on the partially exposed graph. Initially, set $c_i=0$, for every node. Then, until $c_i=k_i$, for all nodes, repeat the following actions. If, for some node $i$, we have $0<c_i<k_i$, then (case~A) select it; otherwise, (case~B) choose freely a node $i$ such that $c_i=0$. Then, in both cases, choose another node $j\neq i$ with probability $P(j)\sim k_j-c_j$. Expose the edge $i\leftrightarrow j$, and increase $c_i$ and $c_j$. Notice that every time we execute case~B, we start the exposition of a new connected component of the graph.

Let $X_r$ be the random variable representing the number of \emph{open vertexes} in partially exposed nodes, i.e. $X_r=\sum_{c_i>0} k_i-c_i$, after the $r$th edge $i\leftrightarrow j$ has been exposed. Notice that we execute case~B when we have $X_{r-1}=0$, and we get $X_r=k_i+k_j-2$. When we execute case~A, there are two situations: (case A1) if node $j$ is a partially exposed node (i.e. $0<c_j$), then $X_r=X_{r-1}-2$, and (case A2) if node $j$ has never been exposed (i.e. $c_j=0$), then $X_r=X_{r-1}+k_j-2$. 

Suppose that cases~B and A1 does not happen very often. Then, the expected change in $X_r$ is
$$
\E[X_r-X_{r-1}] \approx \frac{\sum_{j} k_j (k_j-2)}{\sum_{j} k_j}=\frac{Q(\lambda)}{\E[k]}
$$
and, since $X_r-X_{r-1}\geq -1$, a standard result of random walk theory ensures that if $Q(\lambda)>0$ then, after $\Theta(n)$ steps, $X_r$ is almost surely of order $\Theta(n)$; and if $Q(\lambda)<0$, then $X_r$ returns to zero fairly quickly. In the first case, we generate a giant connected component of size $\Theta(n)$, and in the second case, no component is larger than $\mathcal{O}(\log n)$. In order to prove that executions of case~A1 do not hurt, Molloy and Reed prove that the probability of choosing a partially exposed node (a node with $c_j>0$) is negligible unless we have already exposed a fraction $\Theta(n)$ of the nodes in the current connected component.

Theorems~\ref{thm-criterion} and ~\ref{thm-criterion2} establish a similar criterion for the existence of a giant set of implied literals from a given one. This almost surely implies unsatisfiability of the formula. The proof of the theorems resemble Molloy and Reed's and Cohen et al.'s proofs. In Theorem~\ref{thm-criterion} we fix the number of occurrences of every literal, whereas in Theorem~\ref{thm-criterion2} we fix the number of occurrences of variables. Compared with the definition of $Q$ in Molloy and Reed's, we observe that in Theorem~\ref{thm-criterion2}, the $2$ is replaced by a $3$. In Theorem~\ref{thm-criterion}, we combine the number of literals $k_i$ with the number of their negated $k_{-i}$, and the constant is a $1$ instead of a $2$. Notice that the condition in this case is equal to the condition found by \citet{cooper} for the configuration method and prescribed literal degrees.

\begin{theorem}\label{thm-criterion}
  Let $F$ be a 2-CNF formula generated in a random model with variables $\{x_1,\dots,x_n\}$, where every literal $x_i$ (resp $\neg x_i$) is selected with probability $p_i$ (resp $p_{-i})$, and literals in clauses are not correlated (i.e. $P(x_i\vee x_j)\approx p_i\,p_j$). Assume that $p_i=o(1)$ and $m = \mathcal{O}(n)$. Let $k_i =2\,m\,p_i$ be the expected number of occurrences of literal $x_i$. Then, if $\sum_{i=-n}^n k_i(k_{-i}-1) > 0$, then almost surely $F$ is unsatisfiable.
\end{theorem}

\begin{proof}
  The proof resembles Molloy and Reed's proof for the percolation threshold on graphs. This proof is quite long, and our proof does not differ very much. Therefore, we will only sketch it.
  
  In our case, we do not deal with connected components. In fact, we do not expose the random formula with our algorithm. We assume that we already have the formula, and we describe in Algorithm~\ref{algorithm-exposition} how to enumerates the set of literals implied by a given initial literal $x$.
  \begin{algorithm}[ht]
  \KwIn{$F$,$x$}
  \ForEach{y}{$c_y=0$; $o_y=false$}
  $o_x=true$\;
  \While{$(\exists y.o_{\neg y}=true\wedge c_y<k_y)\wedge \neg (\exists y.o_y=o_{\neg y}=true)$}
  {
  select $y\vee z\in F$ such that $o_{\neg y}=true$\;
  $o_z := true$\label{line-oz}\;
  $c_y:= c_y + 1$\label{line-cy}\;
  $c_z := c_z + 1$\label{line-cz}\;
  $F := F \setminus \{y\vee z\}$
  }
  return $\{z\mid o_z=true\}$
  \caption{Algorithm for finding literals implied by $x$}\label{algorithm-exposition}
  \end{algorithm}
  
  The Boolean variable $o_y$ denotes if the literal has been reached from the initial literal $x$ and the counter $c_y$ denotes the number of clauses containing $y$ that we have already removed from the formula. Therefore, $k_y-c_y$ is the number of clauses containing $y$ that still remains in $F$. When $x$ implies $y$ and $\neg y$, for some variable $y$, we say that $x$ implies a contradiction. In this case, $x$ also implies $\neg x$. The algorithm returns the set of literals implied by $x$ or a contradiction (in this second case, we abort, since we already have $x\to \neg x$ that is what we want to check). Notice also that $c_y>0$ implies $o_y=true\vee o_{\neg y}=true$.
  
  Notice that this algorithm is quite similar to Molloy and Reed's algorithm for exposing connected components of a random graph. Counter $c_y$ has a similar meaning and we only require the Boolean variable $o_y$ to denote the condition of \emph{open vertex} (expressed as $c_y>0$ in Molloy and Reed's algorithm). The algorithm is deterministic, if you consider the formula given. However, for a random formula, the algorithm perform exactly the same steps and can be seen as a random algorithm. Similarly, we can define the random variable $X_r = \sum_{o_{\neg x}=true} k_x - c_x$ after iteration $r$. At every iteration, this variable satisfies:
  \begin{description}
      \item[(case A)] $X_r = X_{r-1} -1$, if $o_z=true$ and $o_{\neg z}=false$,
      \item[(case B)] $X_r = X_{r-1} -1 + k_{\neg z}$, if $o_z=o_{\neg z}=false$, and 
      \item[(case C)] $X_r = X_{r-1} -2 + k_{\neg z} - c_{\neg z}$, if $o_z=false$ and $o_{\neg z}=true$.
  \end{description}
  Notice that line~\ref{line-cy} decreases $X_r$ in $1$, line~\ref{line-cz} decreases $X_r$ in $1$, when $o_{\neg z}=true$, and line~\ref{line-oz} increases $X_r$ in $k_{\neg z}-c_{\neg z}$, when $o_z=false$. However, if both $o_z$ and $o_{\neg z}$ are false, then $c_{\neg z}$ is zero. After case~C, we get a contradiction and finish.
  In case B, the expected gain in $X_r$ is
  \[
  \E[X_r - X_{r-1}] \approx \frac{\sum_z k_z\,(k_{\neg z} -1)}{\sum_z k_z}
  \]
  In the case~A, the random variable only decreases by one. Like Molloy and Reed's, we can also argue that the case~A is negligible, unless we have already added to the set of implied literals a constant fraction of them.
  
  Therefore, reproducing all the lemmas of Molloy and Reed's proof, we can conclude that, when $\sum_z k_z\,(k_{\neg z}-1) >0$, almost surely there exists a constant $0<c<1$ such that for a fraction $c$ of initial literals $x$, the set of literals implied by $x$ is a fraction $c$ of all literals or contains a contradiction, and hence, $X$ implies $\neg x$.
  For a particular variable $x$, the probability that $x$ implies $\neg x$ and $\neg x$ implies $x$ is at least $c^4$. The probability that the formula is satisfiable is at most $(1-c^4)^n$, that tends exponentially to zero as $n$ tends to infinity.
 \end{proof}

\begin{theorem}\label{thm-criterion2}
Let $F$ be a 2-CNF formula generated in a random model with variables $\{x_1,\dots,x_n\}$, where every variable $x_i$ is selected with probability $P_i$ and negated with probability $1/2$, and variables in clauses are not correlated. Assume that $P_i=o(1)$ and $m=\mathcal{O}(n)$. Let $K_i=2\,m\,P_i$ be the expected number of occurrences of variable $x_i$. Then, if $\sum_{i=1}^n K_i(K_i-3) > 0$, then almost surely $F$ is unsatisfiable.

The condition $\sum_{i=1}^n K_i(K_i-3) > 0$ is equivalent to $\E[K^2]/\E[K] > 3$
\end{theorem}

\begin{proof}
  The proof is, like in Theorem~\ref{thm-criterion}, based on \citet{Molloy-Reed}'s proof. In this case, however, the expected gain in the random variable $X_r$ is given by:
  \[
  \E[X_r-X_{r-1}] = \frac{\sum_{i=1}^n \frac{K_i}{2}\left( \frac{K_i-1}{2}-1\right)}{\sum_{i=1}^n \frac{K_i}2}
  \]
  since $\frac{K_i}2$ is the expected value of $k_i$ and $\frac{K_i-1}2$ is the expected value of $k_{-i}$ conditioned to the existence of one positive occurrence of $x_i$.
  Then, the condition $\E[X_r-X_{r-1}]>0$ is equivalent to $\sum_{i=1}^n K_i(K_i-3) > 0$.
  \end{proof}
  
  For the proof of Theorem{\ref{thm-criterion2}, we could also use the \citet{resilence-internet}'s argument. In the case of graphs, we get a giant connected component when a node $i$, connected to a node $j$, is also connected in average to at least one other node.  Formally, when the expected degree of $i$, conditioned to the fact that $i$ and $j$ are connected, is $\E[k_i\mid i\leftrightarrow j]=2$.

  In our case, in order for a giant cycle to emerge, when there is a clause $x\vee y$, we have to find, at least, another clause containing $\neg x$. 
  In this situation, the expected number of other clauses containing $x$ is $2$, that added
  to the original clause $x\vee y$, gives a minimum number of $3$ clauses containing $x$.  
  Given a pair of literals $x$ and $y$, let $\pm x\vee y$ express the fact: ``$x\vee y\in F$ or $\neg x\vee y\in F$. 
  Formally, our criterion can be written as
  \[
  \E[K_x\mid \pm x\vee y] > 3
  \]
  This criterion is the necessary and sufficient condition to continue the construction of a set of clauses, ensuring that the probability that this set contains a fraction of the literals tends to one.
  
  Using Bayes, we have
  \[
    \E[K_x\mid \pm x\vee y]
    = \sum_{k=0}^\infty k\, P(K_x=k\mid  \pm x\vee y)\\
    = \sum_{k=0}^\infty k\, \frac{P(K_x=k\wedge  \pm x\vee y)}{P(\pm x\vee y)}\\
    = \sum_{k=0}^\infty k\, \frac{P(\pm x\vee y\mid K_x=k)P(K_x=k)}{P(\pm x\vee y)}
  \]
  
  Given a pair of literals $x$ and $y$, the probability that either $x\vee y$ or $\neg x\vee y$ are one of the clauses of the formula, conditioned by the fact that the number of occurrences of variable $x$ is $k$ (and assuming that clauses are not repeated) is: 
  \( P(\pm x\vee y\mid K_x=k) =
  \frac{k}{2(n-1)} \) and, the probability of the same fact without condition: \( P(\pm
  x\vee y) = \frac{\E[K_x]}{2(n-1)} \).  
  Therefore
  \[
  \E[K_x\mid \pm x\vee y] 
    = \sum_{k=0}^\infty k\, \frac{\frac{k}{2(n-1)}P(K_x=k)}{\frac{\E[K_x]}{2(n-1)}}\\
    = \frac{\sum_{k=0}^\infty k^2\,P(K_x=k)}{\E[K_x]} = \frac{\E[K^2]}{\E[K]}>3
  \]
  defines an unsatisfiability threshold.
  %

The previous theorems ensure that, when the criterion is satisfied, there is a giant bicycle containing a fraction of the literals, and the formula is unsatisfiable. However, if the formula is unsatisfiable, it can be due to a \emph{small} bicycle. Therefore, the reverse implication is not necessarily true. In other words, Theorems~\ref{thm-criterion} and ~\ref{thm-criterion2} establish a sufficient (but not necessary) condition for unsatisfiability of random 2-SAT formulas, which result into an upper bound for the phase transition point. However, we conjecture that, either giant bicycles are more probable than small bicycles and the percolation threshold (obtained with the criterion) is equal to the phase transition point, or, if small bicycles are more probable, the phase transition point is at $c=0$.

\subsection{Classical 2-SAT Formulas}\label{subsec-classical}

Theorems~\ref{thm-criterion} and~\ref{thm-criterion2} may be used to find the phase transition point in terms of number of clauses divided by number of variables. In this subsection, we apply the technique to (classical) random 2-SAT formulas.

We start with a formula (or graph), not necessarily at the critical threshold. Then, we apply a percolation process where a fraction $1-p$ of randomly selected clauses (edges) are removed, such that the remaining $p$ fraction of edges are in the critical threshold. If we start with the complete formula with all possible $2^2{n\choose 2}$ clauses over $n$ variables, and remove clauses with uniform probability, this process generates a (classical) random 2-SAT formula in the SAT-UNSAT transition point (except for the lack of repeated clauses).

If $k'_x$ is the number of occurrences of literal $x$ in the original graph, then, after removing the $(1-p)$ fraction, the new distribution on the number of occurrences is $P(k_x) = \sum_{k'_x=k_x}^\infty P(k'_x){k'_x\choose k_x}\,p^{k_x}\,(1-p)^{k'_x-k_x}$. Using this binomial distribution we get the moments $\E[k_x]=p\,\E[k'_x]$ and $\E[k_x^2] =
p^2\,\E[(k'_x)^2]+p(1-p)\E[k'_x]$ for any literal $x$. Since $K_x = k_x + k_{\neg x}$, and $k_x$ and $k_{\neg x}$ are independent variables with the same distribution, we have $\E[K_x] = 2\,\E[k_x]$ and $\E[K_x^2]= 2\,\E[k_x^2]+2\,E^2[k_x]$. If we impose the criterion of Theorem~\ref{thm-criterion2} to this new formula we get
\[
\frac{\E[K^2]}{\E[K]}= \frac{2\,p^2\,\E[(k')^2]+2\,p(1-p)\E[k']+2\,p^2\,E^2[k']}{2\,p\,\E[k']}=3
\]
Hence
\[
p=\frac{2}{\frac{\E[(k')^2]}{\E[k']}+\E[k']-1}
\]
For the complete formula we have $k'_x = 2(n-1)$ for any literal, therefore $p=1/(2n-5/2)$. The expected number of clauses in the phase transition threshold is
\[
\E[m]= 2^2{n\choose 2}\,p = \frac{2n(n-1)}{2n-5/2}=n + \mathcal{O}(1)
\]

This proves that the clause/variable fraction at the 2-SAT phase transition threshold is at most $m/n=1$, reproducing the results of Chv\'atal and Reed~\cite{mickgetssome}.

For the expected moments we get $\E[k]\approx 1$ and $\E[k^2]\approx 2$, for number of occurrences of literals, and $\E[K]\approx 2$ and $\E[K^2]\approx 6$, for number of occurrences of variables.

Now, consider the case of (classical) regular random 2-SAT formulas. These are random formulas where the number of occurrences of a literal minus the number of occurrences of another literal is, at most, one. 
Assume that all literals have exactly the same number of occurrences $k_x=m/n$.
Applying Theorem~\ref{thm-criterion}, without any need of percolation process, we get $\sum_{i=-n}^n k_i(k_{-i}-1) = 2n\, \frac m n (\frac m n-1) = 0$. Therefore, $m/n=1$ is an upper bound for the phase transition point, reproducing the results of Boufkhad et al.~\cite{regular-SAT}.
Notice that, in this case, the conditions of Theorem~\ref{thm-criterion2} are not fulfilled: $k_x$ and $k_{\neg x}$ are not independent random variables. If we consider the proof of this Theorem, since in a random regular formula $k_x=k_{\neg x}$, if this formula contains a clause $x\vee y$, we only need to require that $\E[K_x\mid x\vee y]=2$ in order to ensure that there is another clause containing $\neg x$. With this new criterion, and reproducing the proof of Theorem~\ref{thm-criterion2}, we obtain that the threshold in a regular random formula is $\frac{\E[K^2]}{\E[K]} = 2$.

\subsection{Scale-free 2-SAT Formulas}\label{subsec-2SAT}

Recently, Friedrich et al.~\cite{tobiasAAAI17} have proved that scale-free random 2-SAT formulas with exponent $\delta>3$ and clause/variable ratio $m/n<\frac{(\delta-1)(\delta-3)}{(\delta-2)^2}$ are satisfiable with probability $1-o(1)$.\footnote{In their paper, they write $\beta$ instead of $\delta$, but we prefer to use $\beta$ with the same meaning as in~\cite{IJCAI09}.}  This gives a lower bound for a possible phase transition point, in terms of $\delta$. They conjecture that this bound is tight and that this phase transition exists. Replacing $\delta=1/\beta+1$ (according to Theorem~\ref{thm-exponents-relation}) in this inequality, we get:

Scale-free random 2-SAT formulas with exponent $\beta<1/2$ and clause/variable ratio $m/n<\frac{1-2\beta}{(1-\beta)^2}$ are satisfiable with probability $1-o(1)$.

In the first statement of the following theorem, we prove that when the clause/variable ratio exceeds this value, formulas are almost surely unsatisfiable.

\begin{theorem}\label{thm-scale-free-2-SAT}
(1) Scale-free random 2-SAT formulas with exponent $\beta<1/2$ and clause/variable ratio 
\[
m/n>\frac{1-2\beta}{(1-\beta)^2}
\]are unsatisfiable with probability $1-o(1)$.

\noindent (2) Scale-free random 2-SAT formulas over $n$ variables, exponent $\beta=1/2$ and more that 
\[
4\,n\log^{-1}n + \mathcal{O}(n^{1/2}\log^{-1} n)
\]
distinct clauses, or exponent $1/2 < \beta < 1$, and more than
  \[
  \frac{1}{(1-\beta)^2\zeta(2\beta)}n^{2(1-\beta)}  +\mathcal{O}(n^{1-\beta})
  \]
  distinct clauses, are unsatisfiable with probability $1-o(1)$.
\end{theorem}

\begin{proof}
  In the case of scale-free formulas we cannot start the percolation process from the complete formula, since the uniform-random deletion of clauses do not give rise to scale-free formulas. 
  Therefore, we can simply impose the criterion to the original formula.  We will do all the computations using the number of occurrences of variables $K_x$, instead of the number of occurrences of the literal $k_x$, and applying Theorem~\ref{thm-criterion2}. 
  
  Since $\beta<1$, by Lemma~\ref{lem-surname}, repetitions of variables in clauses may be neglected, and the probability that a particular literal in the formula corresponds to the variable $x$ is given by $P(x)\approx\frac{x^{-\beta}}{\sum_{i=1}^n i^{-\beta}}$. Since the election of every variable for every possible literal of the formula is independent, the number of occurrences of $x$ follows a binomial distribution
  \[
  P(K_x\!=\!K) \approx {2m \choose K}\!\!\left(\frac{x^{-\beta}}{\sum_{i=1}^n i^{-\beta}}\right)^K\!\!\!\left(1-\frac{x^{-\beta}}{\sum_{i=1}^n i^{-\beta}}\right)^{2m-K}
  \]
  In the limit $m\to \infty$ the distribution approaches a Poisson distribution where
  \[
  \begin{array}{l}
    \displaystyle \E[K_x]\approx\frac{x^{-\beta}}{\sum_{i=1}^n i^{-\beta}}\,2m\\
    \displaystyle \E[K_x^2]\approx\left(\frac{x^{-\beta}}{\sum_{i=1}^n i^{-\beta}}\,2m\right)^2+\frac{x^{-\beta}}{\sum_{i=1}^n i^{-\beta}}\,2m
  \end{array}
  \]
Recall that in scale-free formulas $K_x$ follows a distinct probability distribution for every variable $x$, therefore we have to average over all variables
  \[
  \begin{array}{ll}
    \E[K] & \displaystyle =\frac 1 n\,\sum_{x=1}^n \E[K_x] \approx 
    \frac 1 n\,\sum_{x=1}^n\frac{x^{-\beta}}{\sum_{i=1}^n i^{-\beta}}\,2m = \frac{2m}{n}
    \\[7mm]
    \E[K^2] & \displaystyle = \frac 1 n\,\sum_{x=1}^n \E[K_x^2] \displaystyle \approx \frac 1 n\,\sum_{x=1}^n\left[\left(\frac{x^{-\beta}}{\sum_{i=1}^n i^{-\beta}}\,2m\right)^2+\frac{x^{-\beta}}{\sum_{i=1}^n i^{-\beta}}\,2m\right]\\[5mm]
    &\displaystyle \approx \frac{4m^2}{n}\ \frac
    {\sum_{x=1}^n x^{-2\beta}}
    {\left[\sum_{i=1} i^{-\beta}\right]^2}+\frac{2m}{n}
  \end{array}
  \]
  
  Imposing the criterion $\E[K^2]/\E[K]=3$ we get
  \[
  m\approx\frac
    {\left[\sum_{i=1}^n i^{-\beta}\right]^2}
    {\sum_{x=1}^n x^{-2\beta}} 
 \]
Applying equations~(\ref{eq-euler}) and~(\ref{eq-euler2}), we get
  \[
  m=\left\{
  \begin{array}{l@{\ \ }l}
    \displaystyle\frac{1-2\beta}{(1-\beta)^2}n  +\mathcal{O}(n^\beta)
    & \mbox{if $\beta<1/2$}\\[4mm]
    \displaystyle 4\,n\log^{-1} n+\mathcal{O}(n^{1/2}\log^{-1} n)
    & \mbox{if $\beta = 1/2$}\\[4mm]
    \displaystyle\frac{1}{(1-\beta)^2\zeta(2\beta)}n^{2(1-\beta)}  +\mathcal{O}(n^{1-\beta})
    & \mbox{if $1/2<\beta<1$}\\[4mm]
    \displaystyle\frac{1}{\zeta(2)}\log^2 n + \mathcal{O}(n^{-1}\log n)
    & \mbox{if $\beta = 1$}\\[4mm]
    \displaystyle\frac{\zeta^2(\beta)}{\zeta(2\beta)}+\mathcal{O}(n^{1-\beta})
    & \mbox{if $1<\beta$}
  \end{array}
  \right.
  \]
  The last two cases are meaningless, since we have assumed that $\beta<1$ in other parts of the proof. The first three possibilities prove the two statements of the theorem. In the second and third case, since we cannot prove that the fraction of repeated clauses is meaningless, we obtain a bound on the number of distinct clauses.
\end{proof}

\begin{figure}[t]
  \centering
  \includegraphics[width=\columnwidth]{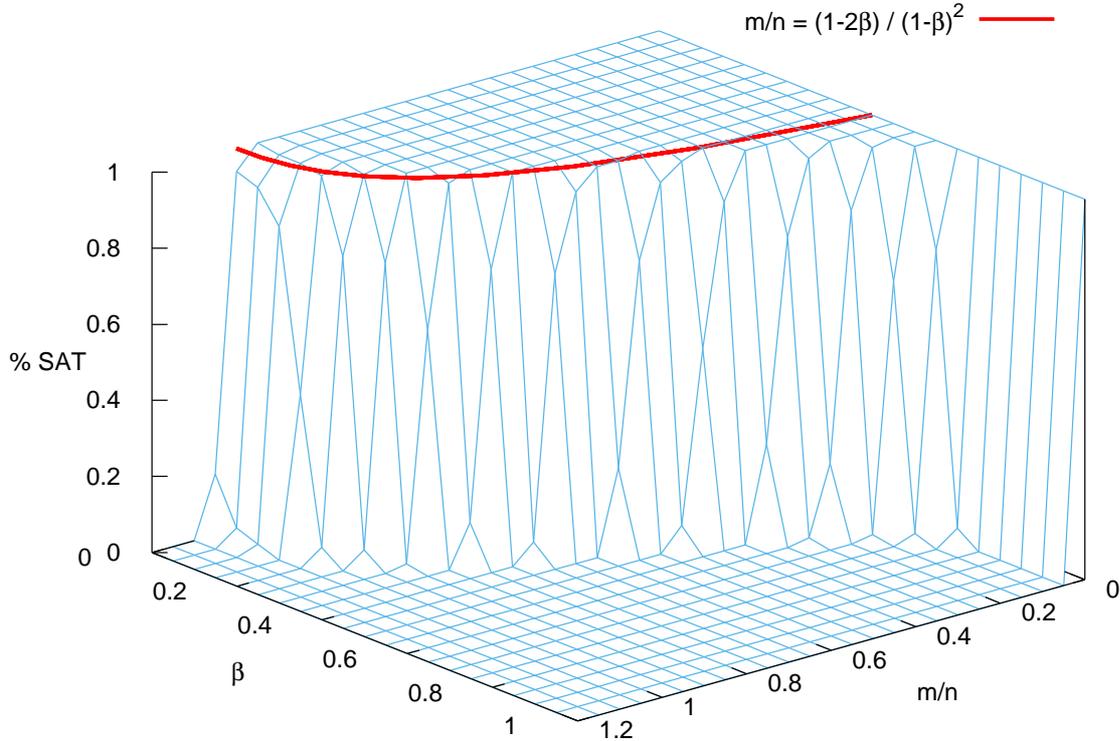}
  \vspace{-10mm}
  
\caption{Fraction of satisfiable formulas as a function of parameter $\beta$ and fraction of clause/variables $m/n$. The number of variables is $n= 10^5$ and the fraction is approximated repeating the experiment for $10$ formulas at every point. We also draw the theoretical threshold $m/n=\frac{1-2\beta}{(1-\beta)^2}$.}
\label{fig-satunsat}
\end{figure}

From~\citet{tobiasAAAI17} and Theorem~\ref{thm-scale-free-2-SAT}
we can conclude:

\begin{corollary}
  Scale-free 2-SAT formulas over $n$ variables and exponent $\beta<1/2$ have a SAT-UNSAT phase transition threshold when the variable/clauses ratio~is
  \[
  m/n = \frac{1-2\beta}{(1-\beta)^2}
  \]
\end{corollary}

We have experimentally analyzed the fraction of satisfiable random scale-free 2-SAT formulas depending on the parameter $\beta$ and fraction of clause/variable $m/n$. The results are plotted in Figure~\ref{fig-satunsat}, for formulas with $n=10^5$ variables. We observe that the phase transition predicted by Theorem~\ref{thm-scale-free-2-SAT} is quite precise, except when $\beta \gtrsim 1/2$. In the limit $n\to \infty$, the fraction of satisfiable formulas with $n$ variables and $c\,n$ clauses tends to zero when $c>0$. However, as the number of clauses needed to make the formula unsatisfiable grows as $n^{2(1-\beta)}$, when $\beta$ is close to $1/2$ the confluence is very slow.

\begin{figure}[t]
  \centering
  \includegraphics[width=\columnwidth]{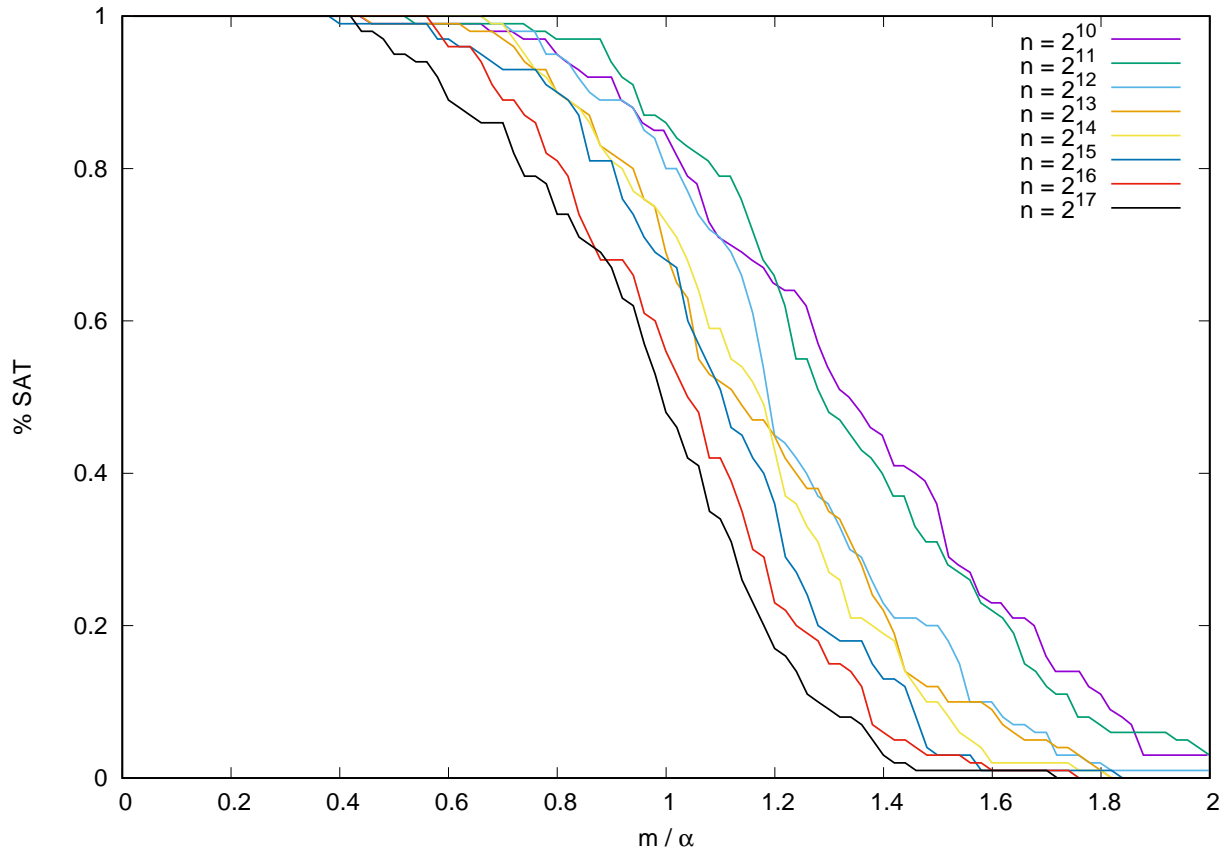}
\caption{Fraction of satisfiable formulas as a function of
  $m/\alpha$, where $\alpha = \frac{1}{(1-\beta)^2\zeta(2\beta)}n^{2(1-\beta)}$ for $\beta=0.7$, and distinct values of $n$ between $2^{10}$ and $2^{17}$. Every point is computed repeating the experiment for $100$ formulas, and checking how many of them are satisfiable.}
  \label{fig-varn}
\end{figure}

In order to test experimentally the second statement of Theorem~\ref{thm-scale-free-2-SAT}, we have analyzed the fraction of satisfiable formulas with respect to $m/\alpha$, where $\alpha = \frac{1}{(1-\beta)^2\zeta(2\beta)}n^{2(1-\beta)}$. In Figure~\ref{fig-varn}, we show the results for $\beta=0.7$. We observe that, for distinct values of $n$, the transition between SAT and UNSAT is around $\alpha$. However, for increasing values of $n$ the transition does not seem to become more abrupt.

\section{Unsatisfiability by Small Cores}\label{sec-small-cores}

In the proof of Theorem~\ref{thm-scale-free-2-SAT} we have already seen that, when $\beta>1/2$ the number of clauses needed to make a 2-SAT formula unsatisfiable is sub-linear. Therefore, the phase transition factor --understood as a constant $c$ such that, on the limit $n\to \infty$, formulas with less that $c\,n$ clauses are satisfiable and those with more than $c\,n$ clauses are unsatisfiable-- is zero. In this section, we will prove that, when $\beta$ exceeds a certain bound, scale-free formulas become unsatisfiable due to a small subset of clauses containing variables with small indexes. Moreover, this result holds for clauses of any size.

\begin{theorem}\label{thm-exponent}
  A random scale-free formula over $n$ variables, exponent $0\leq \beta<1$ and
  $\omega(n^{(1-\beta)k})$ clauses of size $k$ is unsatisfiable with
  probability $1-o(1)$.
\end{theorem}
\begin{proof}
  The probability of a clause only containing the smallest $k$ variables is
  \[
  P(x_1\vee\dots\vee x_k) \approx \frac{k!\,1^{-\beta}\cdots k^{-\beta}}{\left(2\, \sum_{i=1}^n i^{-\beta}\right)^k}
  \]
  This inequality would be an equality, if we allow tautologies and simplifiable clauses (i.e. repeated variables) in formulas.

  Using (\ref{eq-euler}) we get
  \[
  P(x_1\vee\dots\vee x_k)  \displaystyle \approx
  \frac{k!\,1^{-\beta}\cdots k^{-\beta}}{\left( 2\,\sum_{i=1}^n i^{-\beta}\right)^k} =
  \frac{(k!)^{1-\beta}}
  {2^k\,\left(\zeta(\beta)+\frac{n^{1-\beta}}{1-\beta}+\mathcal{O}(n^{-\beta})\right)^k}
  \]

In the limit $n\to \infty$, the probability of generating the clause $x_1\vee\dots\vee x_k$ after generating $n^{(1-\beta)k}$ independent clauses is
\[
1 - \left( 1 - \frac{(k!)^{1-\beta}}{2^k\, \left(\frac{n^{1-\beta}}{1-\beta}\right)^k}\right)^{n^{(1-\beta)k}} =
1 - \left( 1 - \frac{\left(\frac{1-\beta}{2}\right)^k\,(k!)^{1-\beta}}{n^{(1-\beta)k}}\right)^{n^{(1-\beta)k}} \approx 1-e^{-\left(\frac{1-\beta}{2}\right)^k(k!)^{1-\beta}}
\]
Therefore, the probability of generating the clause $x_1\vee\dots\vee x_k$ is $1-o(1)$ when the number of clauses is $m=\omega(n^{(1-\beta)k})$. The same applies for other $2^k$ clauses  with distinct signs, and, if $k=\mathcal{O}(1)$, to a refutation of the formula only using these set of clauses.
\end{proof}

As in classical random formulas, the expected number of truth assignments that satisfy a scale-free random formula is $2^n(1-2^{-k})^m$. This imposes a linear upper bound on the number of clauses of satisfiable scale-free formulas, i.e. a random scale-free formula with $m=c\,n$ clauses of size $k$ over $n$ variables such that $c>2^k\,\log2$ is unsatisfiable with probability $1-o(1)$. Therefore, the bound in Theorem~\ref{thm-exponent} only \emph{improves} this other linear bound when $(1-\beta)k <1$, hence when $\beta>1-1/k$.

Figure~\ref{fig-transition} shows an experimental estimation of how many clauses are needed to make unsatisfiable $50\%$ of the random formulas generated with distinct values of $\beta$ and $k=3$, as a function of the number of variables. 

\begin{figure}[ht]
\begin{center}
\includegraphics{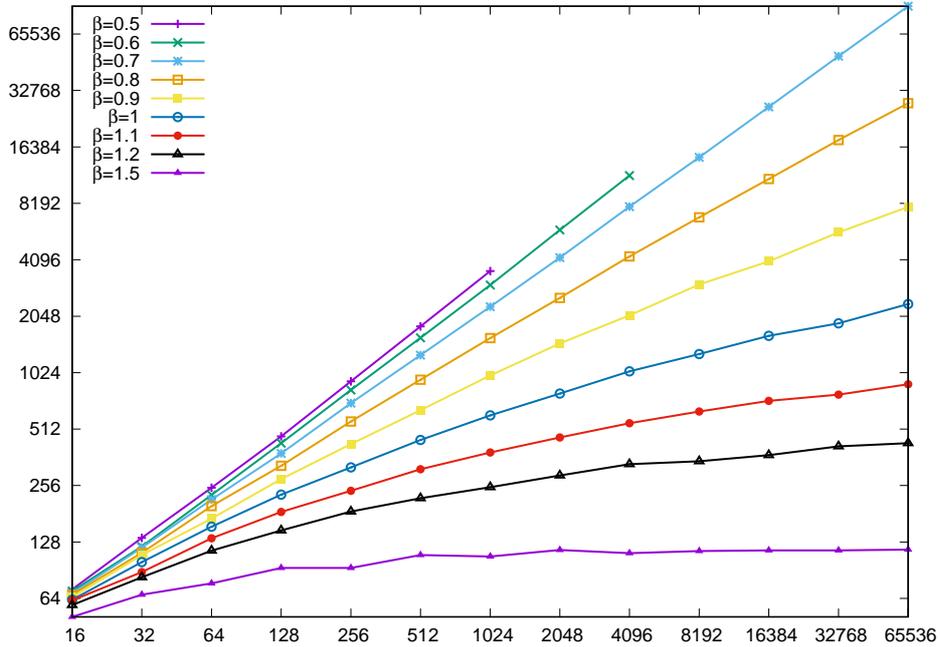}
\end{center}
\caption{Estimation of the number of clauses that are needed to make unsatisfiable $50\%$ of the formulas generated for distinct values of $\beta$ and $k=3$ as a function of the number of variables.}\label{fig-transition}
\end{figure}

Theorem~\ref{thm-exponent} predicts that the number of clauses in a satisfiable scale-free 2-SAT formula cannot grow faster that $\mathcal{O}(n^{2(1-\beta)})$, due to the emergence of small cores. When $1/2<\beta<1$, the second statement of Theorem~\ref{thm-scale-free-2-SAT}, predicts exactly the same exponent $2(1-\beta)$ for the emergence of a giant bicycle. This suggest that, in this range of $\beta$, the probability of existence of a small and a giant unsatisfiable core of clauses is similar. However, experimental results (see Figure~\ref{fig-varn}) suggest that the SAT-UNSAT transition is quite smooth, like in classical 1-SAT. This suggest that small cores are, in fact, more prominent. Another argument in this direction is as follows:

Let $C(V)$ be the subset of clauses only containing variables of the subset $V$ of variables. The greater $|C(V)|/|V|$ is, the higher is the probability to have an unsatisfiable core inside $C(V)$. In the case of scale-free random k-SAT formulas, let $C_r$ be the set of clauses only containing variables $\{1,\dots,r\}$. We can estimate
\[
E\left[\frac{|C_r|}{r}\right] = \frac mr\left(\frac{\sum_{i=1}^r i^{-\beta}}{\sum_{i=1}^n i^{-\beta}}\right)^k\approx
\frac mr\left(\frac{r^{1-\beta}-1}{n^{1-\beta}-1}\right)^k
\]
For $(1-\beta)k \geq 1$, i.e. $\beta < 1-1/k$, the maximum of this function is $r=\infty$.  For $(1-\beta)k < 1$, i.e. $\beta > 1-1/k$, the maximum is finite:
\[
r= (1\!-\!(1\!-\!\beta)k)^{-1/(1-\beta)}
\]
Notice that $(1-\beta)k$ is the exponent predicted by Theorem~\ref{thm-exponent}, and that for 2-SAT, $1-1/k=1/2$. Therefore, we get another proof that at $\beta=1-1/k$ we get a change in the behavior of scale-free random k-SAT formulas. When $n\to\infty$, for $\beta \leq 1-1/k$ the most probable is to get a very large core that involves a fraction of the whole set of clauses. For $\beta > 1-1/k$ the most probable is to get a small core only involving a finite set of clauses and variables $\{1,\dots,(1\!-\!(1\!-\!\beta)k)^{-1/(1-\beta)}\}$.

\section{Conclusions}

We have proposed a new model of generation of random SAT formulas that mimic better the properties observed in \emph{real world} formulas. In particular, the number of occurrences of variables follow a power-law distribution, as observed in the industrial SAT instances used in competitions. This is obtained by assigning a distinct probability $P(i)\sim i^{-\beta}$ to every variable $i\in\{1,\dots,n\}$, where $\beta$ is a parameter. This model generalizes (classical) random SAT formulas by taking $\beta=0$.

We prove the existence of a SAT-UNSAT phase transition for 2-CNF formulas. This result is obtained using a novel technique based on percolation techniques. For arbitrary k-CNF formulas, we prove that formulas with more than $\omega(n^{(1-\beta)k})$ clauses are unsatisfiable with probability $1-o(1)$. More precisely,  that when $\beta>1-1/k$ formulas are unsatisfiable due to a small set of clauses that only involve most frequent variables.

\bibliographystyle{elsart-harv.bst}
\bibliography{aij5}

\end{document}